\documentclass[12pt]{iopart}
\usepackage{amssymb,amsthm,accents}

\eqnobysec
\newtheorem{prop}{Proposition}
\newtheorem{definition}{Definition}

\newcommand{\wb}{\bar}
\newcommand{\wt}{\tilde}
\newcommand{\wh}{\hat}
\newcommand{\ws}{\dot}

\newcommand{\us}[1]{\underaccent{\dot}{{#1}}}
\renewcommand{\th}[1]{\wh{\wt{#1}}}
\newcommand{\hb}[1]{\wb{\wh{#1}}}
\newcommand{\bt}[1]{\wt{\wb{#1}}}

\newcommand{\cP}{\mathcal{P}}
\newcommand{\cQ}{\mathcal{Q}}
\newcommand{\cR}{\mathcal{R}}
\newcommand{\cB}{\mathcal{B}}
\newcommand{\cH}{\mathcal{H}}
\newcommand{\cF}{{f}}

\newcommand{\mt}{\mathsf{M}}
\newcommand{\mm}{\mathsf{m}}
\renewcommand{\mr}{\mathsf{r}}
\renewcommand{\ms}{\mathsf{s}}

\newcommand{\dd}{*}
\newcommand{\ecp}{{\mathbb{C}\cup\{\infty\}}}
\newcommand{\sn}{\mathop{\mathrm{sn}}\nolimits}

\newcommand{\CR}[4]{{\frac{({#1}-{#2})({#3}-{#4})}{({#1}-{#3})({#2}-{#4})}}}

\begin{document}
\paper{The Schwarzian variable associated with discrete KdV-type equations}
\author{James Atkinson and Nalini Joshi}
\address{School of Mathematics and Statistics, The University of Sydney, NSW 2006, Australia.}
\date{\today}
\begin{abstract}
We present a new construction related to systems of polynomials which are consistent on a cube.
The consistent polynomials underlie the integrability of discrete counterparts of integrable partial differential equations of Korteweg-de Vries-type (KdV-type).
The construction reported here associates a Schwarzian variable to such systems.
In the generic case, including the primary model Q4, the new variable satisfies the lattice Schwarzian Kadomtsev-Petviashvili (KP) equation in three dimensions. 
For the degenerate sub-cases of Q4 the same construction reveals an invertible transformation to the lattice Schwarzian KdV equation.
\end{abstract}

\section{Introduction}\label{INTRO}
The feature of M\"obius point symmetry group of an equation is synonymous with the presence of this group's differential invariant, the {\it Schwarzian derivative}, or in the discrete setting the {\it cross-ratio}.
Schwarzian forms of integrable partial differential equations appeared first in the work of Weiss through the Painlev\'e analysis \cite{wei} as the {\it singularity manifold equation} (cf. \cite{wtc}).
For a review of the continuous and discrete Schwarzian integrable systems see \cite{nij}.

The discrete KdV-type equations were classified using the multidimensional consistency integrability criterion in pioneering work of Adler, Bobenko and Suris (ABS) in \cite{abs1,abs2} (cf. \cite{hie,atk2,bol}).
The purpose of this paper is to introduce a new variable associated with these equations. 
It is defined up to the action of the M\"obius group, and is therefore characterised by Schwarzian systems.
We will focus on the contribution this new definition makes to the transformation theory of these equations.

It can be useful to partition the ABS list of equations into two parts. In the first part we put equations which are equivalent to the older NQC equation \cite{nqc} and its parameter sub-cases, namely Q3$_{\delta=0}$, Q1, A2, A1, H3$_{\delta=0}$ and H1 in the nomenclature of \cite{abs1}.
Although all ABS equations share the core characteristic of multidimensional consistency, there are several features of the NQC-type equations which do not directly transfer to the second part of the ABS list.
One such feature is the connection to a discrete KP-type equation. 
Inherent from the author's approach the NQC equation was connected with the discrete KP-type equation obtained in \cite{ncwq}, specifically the natural embedding of the former equation in three dimensions satisfies the latter.
This fulfils the expectation of dimensional reduction that is familiar from the continuous theory. 
It is an important feature; revealing the wider natural context for the integrable systems.
For the second and newer part of the ABS list, namely equations Q4, Q3$_{\delta=1}$, Q2, H3$_{\delta=1}$ and H2, the connection with a discrete KP-type equation should be expected, but does not appear in the direct way as for the NQC case.
This is resolved on the level of the variable introduced here, because for all ABS equations the new variable satisfies the lattice Schwarzian KP equation in three dimensions.

The Schwarzian construction also enables us to systematically connect all equations on the ABS list, with the exception of Q4, to the lattice Schwarzian KdV equation.
This complements several relationships of Miura and B\"acklund type which are already known \cite{nc,atk1,nah}.
It also shows that allowing for this class of transformations there are at most two distinct ABS equations, which provides a discrete counterpart to the result of \cite{ssy}.
A virtue of the transformation to the Schwarzian variable is a generalised transitivity property which means the transformation inherently captures also the {\it difference} between the equations which it connects.

The paper is organised as follows.
Many elements of the construction are introduced on the level of the Riccati equation in Section \ref{ODE}.
Section \ref{ABT} contains the definition of the Schwarzian variable in the context of the Riccati equations associated with the B\"acklund transformations of the discrete KdV-type equations. 
The definition is extended to higher dimensions in Section \ref{3D}.
The Q4 situation is considered in Section \ref{ADLER} and the remaining systems of \cite{abs1} in Section \ref{DEGENS}.
In the final section we describe what happens when we apply the definition to systems of consistent polynomials less symmetric than those of \cite{abs1} by considering several examples.

\section{A Schwarzian variable associated with the Riccati map}\label{ODE}
Consider the scalar discrete Riccati equation
\begin{equation}
\cR(v,\wt{v}) := a_0 + a_1 v + a_2 \wt{v} + a_3 v \wt{v} = 0,\label{gr}
\end{equation}
where $v=v(n)$ and $\wt{v}=v(n+1)$ are values of the dependent variable $v : \mathbb{Z}\rightarrow \ecp$ as a function of the independent variable $n\in\mathbb{Z}$, and $a_0\ldots a_3 : \mathbb{Z}\rightarrow\mathbb{C}$ are non-autonomous coefficients of the polynomial $\cR$.
The values of the variable $v$ are naturally associated with the vertices of a one-dimensional lattice, and the polynomial $\cR$ connects consecutive vertices and so is naturally associated with the edges.
A degenerate situation occurs if the polynomial on some edge is reducible, we will assume explicitly that this is not the case, specifically that
\begin{equation}
(\partial_1\cR)(\partial_2\cR)-(\partial_1\partial_2\cR)\cR = a_1a_2-a_0a_3\neq 0,\label{redR}
\end{equation}
where we have introduced the convenient notation $\partial_i$ for partial differentiation of the polynomial with respect to its $i^{th}$ argument.

There is a close relationship between the Riccati equation and the group of M\"obius transformations,
\begin{equation}
\mt := \left\{ x\mapsto(ax+b)/(cx+d) : a,b,c,d\in \mathbb{C}, ad\neq bc \right\}.
\end{equation}
This is more discernible if we reformulate (\ref{gr}) as an equation for a function $\ms:\mathbb{Z}\rightarrow\mt$,
\begin{equation}
\wt{\ms} = \mr\cdot \ms,\label{mr}
\end{equation}
where $\mr=x\mapsto-(a_1x+a_0)/(a_3x+a_2)$.
The values of the functions $\ms$ and $\mr$ are naturally associated to the vertices and edges respectively of the one-dimensional lattice.
The virtue of the reformulation (\ref{mr}) is that both of these functions take values in the group $\mt$.
The equations (\ref{gr}) and (\ref{mr}) are equivalent in the sense that from the general solution of one we may construct the general solution of the other. 
This is clear because given $\ms$ satisfying (\ref{mr}), $\ms(c)$ satisfies (\ref{gr}) for any constant $c$.
And conversely, if the general solution of (\ref{gr}) is known, then certainly three particular solutions are known, say $v_1$, $v_2$ and $v_3$, and writing $\ms(1)=v_1$, $\ms(\infty)=v_2$ and $\ms(0)=v_3$ determines uniquely a function $\ms$ satisfying (\ref{mr}). 
Note that constants other than $1$, $\infty$ and $0$ would also be fine, provided they were distinct, however this choice results in the convenient expression
\begin{equation}
\ms^{-1}= x\mapsto \CR{v_1}{v_2}{v_3}{x}.\label{invs}
\end{equation}
By $\ms^{-1}$ we mean the unique function $\ms^{-1}:\mathbb{Z}\rightarrow\mt$ taking the value in $\mt$ which is the group inverse of the value of $\ms$.

We remark that letting $v$ take values in $\ecp$ is now more obviously natural because elements of $\mt$ permute this set.

We now introduce the object of principal concern, a function $\varphi:\mathbb{Z}\rightarrow\ecp$ associated with (\ref{gr}) which is defined in terms of a fixed but as yet un-specified function $w:\mathbb{Z}\rightarrow\ecp$ as
\begin{equation}
\varphi:=\ms^{-1}(w)\label{phidef},
\end{equation}
where $\ms$ is a solution of (\ref{mr}).
The solution of (\ref{mr}) is unique up to the transformation $\ms\rightarrow\ms\cdot\mm$ for arbitrary $\mm\in\mt$, and therefore the definition (\ref{phidef}) determines $\varphi$ up to the transformation $\varphi\rightarrow\mm(\varphi)$.
This connects the definition (\ref{phidef}) with the following one-dimensional Schwarzian difference equation:
\begin{equation}
\CR{\varphi}{\wt{\varphi}}{\wt{\wt{\varphi}}}{\wt{\wt{\wt{\varphi}}}}=
\CR{[\wt{\wt{\mr}}\cdot\wt{\mr}\cdot\mr](w)}{[\wt{\wt{\mr}}\cdot\wt{\mr}](\wt{w})}{\wt{\wt{\mr}}(\wt{\wt{w}})}{\wt{\wt{\wt{w}}}}.\label{seq}
\end{equation}
Here we have simply evaluated the cross-ratio of four consecutive values of $\varphi$, the right-hand-side follows by substituting for $\varphi$ from (\ref{phidef}) and using (\ref{mr}) together with the M\"obius invariance of the cross-ratio, therefore (\ref{seq}) is a consequence of our definition.
However this equation actually {\it characterises} $\varphi$ because any two solutions of (\ref{seq}) are M\"obius related: observe first, the equation is invariant under M\"obius transformations on $\varphi$, and second, the equation is third order, which ensures the existence of a M\"obius transformation sending any set of well-posed initial data to any other.

Notice that the right-hand-side of (\ref{seq}) does not depend on $\ms$, so solving this equation obtains $\varphi$ without requiring first the solution of (\ref{gr}).
There is however the following subtlety:
from the definition (\ref{phidef}) it is easy to see that $\varphi$ is constant if and only if $w$ is a solution of (\ref{gr}), whereas the equation (\ref{seq}) becomes ambiguous in this situation. 
So definition (\ref{phidef}) is the more robust.

We remark that the $\varphi$ variable can play a useful role when only two particular solutions of (\ref{gr}) are known.
If we label the known solutions $v_2$ and $v_3$ then the associated variable $\varphi$ may be found by integration of the homogeneous linear equation
\begin{equation}
(v_3-w)\cR(v_2,\wt{w})\wt{\varphi}-(v_2-w)\cR(v_3,\wt{w})\varphi = 0.\label{hle}
\end{equation}
This equation is the result of combining (\ref{invs}) and (\ref{phidef}) and demanding $v_1$ also satisfy (\ref{gr}).
Solving (\ref{hle}) obtains $\varphi$ from $v_2$ and $v_3$.
Once a particular solution of (\ref{hle}) is known (excluding $\varphi=0$ and $\varphi=\infty$) we can again combine (\ref{invs}) and (\ref{phidef}) to reconstruct solution $v_1$, and thus the general solution of (\ref{gr}). 
Therefore if two particular solutions are known, the transformation to the variable $\varphi$ reduces (\ref{gr}) to the homogeneous linear equation (\ref{hle})\footnote{This transformation is closely related to the one exploited in \cite{an2}, which corresponds to the choice $w=\infty$.}.

To conclude this preliminary section let us re-iterate that our intention so far has been to establish, separately from further considerations, a basic definition of $\varphi$, its M\"obius invariance, as well as some explicit formulae.
This is to emphasise aspects of the construction which are present already at the level of the Riccati equation.
They can be distinguished from features like the choice of $w$, a function which has thus far played a rather passive role but which is important later.

\section{Application to B\"acklund transformations}\label{ABT}
Well known B\"acklund transformations for discrete KdV-type equations are of the following generic form
\begin{equation}
\cB_1(u,\wt{u},v,\wt{v})=0, \quad \cB_2(u,\wh{u},v,\wh{v})=0.\label{gbt}
\end{equation}
Here $\cB_1$ and $\cB_2$ are autonomous polynomials of degree one in four variables, $u,v:\mathbb{Z}^2\rightarrow\ecp$ are dependent variables and $\wt{\phantom{u}}, \wh{\phantom{u}}$ denote shifts on these variables in the two lattice directions.
Such B\"acklund transformations connect quadrilateral lattice equations in $u$ and $v$,
\begin{equation}
\cP(u,\wt{u},\wh{u},\th{u})=0, \quad  \cP^\dd(v,\wt{v},\wh{v},\th{v})=0,\label{geqs}
\end{equation}
where $\cP$ and $\cP^\dd$ are also autonomous polynomials of degree one in four variables. 
The key property we assume is the consistency of the polynomials when associated to the faces of a cube as described in \cite{nw,bs1} - the polynomials $\cP$ and $\cP^\dd$ here lie on opposite faces.
Choosing $u$ to satisfy the equation on the left in (\ref{geqs}) results in the compatibility of (\ref{gbt}) as a system for $v$, and the function $v$ which emerges as the solution of this system then satisfies the equation on the right in (\ref{geqs}).

In the theory of such systems developed in \cite{abs1,abs2} the following polynomials play a fundamental role:
\begin{equation}
\fl \quad
\begin{array}{ll}
\cH_1:=(\partial_3\cB_1)(\partial_4\cB_1)-(\partial_3\partial_4\cB_1)\cB_1,\quad &
\cH_2:=(\partial_3\cB_2)(\partial_4\cB_2)-(\partial_3\partial_4\cB_2)\cB_2,\\
\cH_1^\dd:=(\partial_1\cB_1)(\partial_2\cB_1)-(\partial_1\partial_2\cB_1)\cB_1, &
\cH_2^\dd:=(\partial_1\cB_2)(\partial_2\cB_2)-(\partial_1\partial_2\cB_2)\cB_2.
\end{array}
\label{biquad}
\end{equation}
The first two of these arise in the non-degeneracy condition for (\ref{gbt}) as Riccati equations for $v$:
\begin{equation}
\cH_1(u,\wt{u})\neq 0, \quad \cH_2(u,\wh{u})\neq 0.\label{ndc}
\end{equation}
If these conditions hold the solution $u$ is said to be non-singular\footnote{The definition of singularity here, which is a singularity {\it of the B\"acklund transformation}, differs from the definition of singularity {\it of the equation} in \cite{abs2}, however for most examples these notions coincide.}.

If $u$ is a non-singular solution of the equation on the left in (\ref{geqs}), then the Riccati equations for $v$ in (\ref{gbt}) can be reformulated as the equivalent system
\begin{equation}
\wt{\ms}=\mr_1\cdot\ms, \quad \wh{\ms}=\mr_2\cdot\ms,\label{ssys}
\end{equation}
for a function $\ms:\mathbb{Z}^2\rightarrow\mt$. 
Here $\mr_1,\mr_2:\mathbb{Z}^2\rightarrow\mt$ are the M\"obius transformations $v\mapsto\wt{v}$ and $v\mapsto\wh{v}$ defined by (\ref{gbt}), which are naturally associated to the edges of the quadrilateral lattice.
The values of $\ms$ are associated to the vertices.
Compatibility of (\ref{ssys}) means simply that $\wh{\mr}_1\cdot\mr_2=\wt{\mr}_2\cdot\mr_1$.
This reformulation allows the definition of a Schwarzian variable $\varphi$ associated with a non-singular solution of the equation defined by $\cP$:
\begin{definition}\label{mudef}
Let $w:\mathbb{Z}^2\rightarrow \ecp$ be a fixed solution of the system
\begin{equation}
\cH_1^\dd(w,\wt{w})=0, \quad \cH_2^\dd(w,\wh{w})=0,\label{wdef}
\end{equation}
and $u:\mathbb{Z}^2\rightarrow \ecp$ be a non-singular solution of the equation on the left in (\ref{geqs}).
We refer to $\varphi$ as the Schwarzian variable associated to $u$ if $\varphi=\ms^{-1}(w)$ for some $\ms:\mathbb{Z}^2\rightarrow\mt$ satisfying the B\"acklund system (\ref{ssys}).
We denote the set of all such pairs of functions $(u,\varphi)$ by $\mu$.
\end{definition}
The solution of (\ref{ssys}) is unique up to the transformation $\ms\rightarrow\ms\cdot \mm$ for arbitrary $\mm\in\mt$.
Thus if $(u,\varphi)\in\mu$, then $(u,\varphi')\in\mu$ if and only if there exists $\mm\in\mt$ such that $\varphi'=\mm(\varphi)$.
Although it is immediate from the definition, this is a key property of the relation $\mu$, in particular it has the consequence that
\begin{equation}
(u',\varphi),(u,\varphi),(u,\varphi')\in\mu \ \Rightarrow \ (u',\varphi')\in\mu.\label{muprop}
\end{equation}
Property (\ref{muprop}) is a generalisation of transitivity to the situation where the domain of the relation is different than the co-domain.
It induces a partition on both domain and co-domain, and the relation lifts to a bijection between the equivalence classes.
In the co-domain this is simply partition by M\"obius equivalence, whereas the partition of the domain reveals a kind of symmetry of the equation defined by polynomial $\cP$.

Based on the established M\"obius invariance of the co-domain of $\mu$, we will further investigate this set of functions using the cross-ratio.
The cross-ratio of four values of $\varphi$ around an elementary quadrilateral is found to be
\begin{equation}
\CR{\varphi}{\wt{\varphi}}{\wh{\varphi}}{\th{\varphi}}=
\frac{([\wt{\mr}_2\cdot\mr_1](w)-\wt{\mr}_2(\wt{w}))(\wh{\mr}_1(\wh{w})-\th{w})}{([\wh{\mr}_1\cdot\mr_2](w)-\wh{\mr}_1(\wh{w}))(\wt{\mr}_2(\wt{w})-\th{w})}\label{s2d}
\end{equation}
immediately from Definition \ref{mudef}.
The right-hand-side of (\ref{s2d}) can be evaluated for particular examples of systems (\ref{gbt}), and examining the result will lead to insights about the functions in the co-domain of $\mu$.
More generally, M\"obius-invariant constraints on $\varphi$ can be obtained by considering the cross-ratio of values on any set of four distinct vertices of the lattice, the vertices of a single quadrilateral are the most primitive such set.
We remark that (\ref{s2d}) evaluates in the generic case, and without any assumption on $w$, to 
\begin{equation}
\CR{\varphi}{\wt{\varphi}}{\wh{\varphi}}{\th{\varphi}}=
\chi\frac{\cB_1(u,\wt{u},w,\wt{w})\cB_1(\wh{u},\th{u},\wh{w},\th{w})}{\cB_2(u,\wh{u},w,\wh{w})\cB_2(\wt{u},\th{u},\wt{w},\th{w})},
\label{s2de}
\end{equation}
where $\chi$ is independent of $w$.
$\chi$ is easily found as a rational expression in the coefficients of $\mr_1$, $\mr_2$, $\wh{\mr}_1$ and $\wt{\mr}_2$.
There are several ways to express it because the coefficients are themselves related due to the compatibility $\wh{\mr}_1\cdot\mr_2=\wt{\mr}_2\cdot\mr_1$.
Interestingly only its square seems to be expressible symmetrically, and is found to be
\begin{equation}
\chi^2=\frac{\cH_2(u,\wh{u})\cH_2(\wt{u},\th{u})}{\cH_1(u,\wt{u})\cH_1(\wh{u},\th{u})}.
\end{equation}

The choice of the fixed function $w$ is an important element of Definition \ref{mudef}\footnote{Generically $w$ satisfying (\ref{wdef}) is a solution of the equation on the right in (\ref{geqs}) which is singular throughout the lattice.}, it results in the reducibility of the four terms on the right-hand-side of (\ref{s2de}) when they are considered as polynomials in the four variables $u,\wt{u},\wh{u},\th{u}$. 
Note that the analog of this in the generic one-dimensional situation is not clear; choosing $w$ in this way is a possibility particular to the Riccati equations defined in terms of the consistent polynomials.

\section{Extension to higher dimensions}\label{3D}
In most cases the B\"acklund transformation described in the previous section embeds naturally in a three-dimensional system:
\begin{equation}
\cB_1(u,\wt{u},v,\wt{v})=0, \quad \cB_2(u,\wh{u},v,\wh{v})=0,\quad \cB_3(u,\wb{u},v,\wb{v})=0,\label{gbt3D}
\end{equation}
where now $u,v:\mathbb{Z}^3\rightarrow \ecp$ and $\wt{\phantom{u}}$, $\wh{\phantom{u}}$, $\wb{\phantom{u}}$ denote shifts in the three lattice directions.
Such extended B\"acklund transformations connect systems in $u$ and $v$:
\begin{equation}
\begin{array}{ll}
\cP_{12}(u,\wt{u},\wh{u},\th{u})=0,\qquad & \cP_{12}^\dd(v,\wt{v},\wh{v},\th{v})=0,\\
\cP_{23}(u,\wh{u},\wb{u},\hb{u})=0,& \cP_{23}^\dd(v,\wh{v},\wb{v},\hb{v})=0,\\
\cP_{31}(u,\wb{u},\wt{u},\bt{u})=0,& \cP_{31}^\dd(v,\wb{v},\wt{v},\bt{v})=0,
\end{array}
\label{geqs3D}
\end{equation}
where $\cP_{12}$ and $\cP_{12}^\dd$ are just a relabelling of $\cP$ and $\cP^\dd$ from before.
The property required is the consistency of all these polynomials on a four-dimensional hypercube.

The generalisation of the Schwarzian variable to this higher dimensional situation is fairly obvious, we give the details here only for clarity. 
The generalised definition will be important in order to establish the connection between the two-dimensional systems and the three-dimensional lattice Schwarzian KP equation.

We define polynomials $\cH_3$ and $\cH_3^\dd$ in the natural way extending (\ref{biquad}),
\begin{equation}
\eqalign{
\cH_3:=(\partial_3\cB_3)(\partial_4\cB_3)-(\partial_3\partial_4\cB_3)\cB_3,\\
\cH_3^\dd:=(\partial_1\cB_3)(\partial_2\cB_3)-(\partial_1\partial_2\cB_3)\cB_3.
}
\label{biquad3D}
\end{equation}
A solution of the system on the left in (\ref{geqs3D}) is then said to be non-singular if
\begin{equation}
\cH_1(u,\wt{u})\neq 0, \quad \cH_2(u,\wh{u})\neq 0, \quad \cH_3(u,\wb{u})\neq 0\label{ndc3D}
\end{equation}
throughout $\mathbb{Z}^3$.
If $u$ is a non-singular solution of the system on the left in (\ref{geqs3D}), then the equations (\ref{gbt3D}) for $v$ can be reformulated as the equivalent system
\begin{equation}
\wt{\ms}=\mr_1\cdot\ms, \quad \wh{\ms}=\mr_2\cdot\ms, \quad \wb{\ms}=\mr_3\cdot\ms \label{ssys3D}
\end{equation}
for a function $\ms:\mathbb{Z}^3\rightarrow\mt$.
This reformulation of the B\"acklund equations leads naturally to the generalised definition of the Schwarzian variable:
\begin{definition}\label{mudef3D}
Let $w:\mathbb{Z}^3\rightarrow \ecp$ be a fixed solution of the system
\begin{equation}
\cH_1^\dd(w,\wt{w})=0, \quad \cH_2^\dd(w,\wh{w})=0, \quad \cH_3^\dd(w,\wb{w})=0,\label{wdef3D}
\end{equation}
and $u:\mathbb{Z}^3\rightarrow \ecp$ be a non-singular solution of the system on the left in (\ref{geqs3D}).
We refer to $\varphi$ as the Schwarzian variable associated to $u$ if $\varphi=\ms^{-1}(w)$ for some $\ms:\mathbb{Z}^3\rightarrow\mt$ satisfying the B\"acklund system (\ref{ssys3D}).
\end{definition}
{\noindent The extension to higher dimensions is clear. In other words the Schwarzian-variable construction is compatible with the multidimensional consistency.}

\section{The Q4 Schwarzian variable}\label{ADLER}
In \cite{abs2} it is shown that the generic non-degenerate compatible system of polynomials (\ref{gbt}), (\ref{geqs}) may be taken, without loss of generality, in the form 
\begin{equation}
\cB_1=\cQ_{\alpha,\kappa},\quad\cB_2=\cQ_{\beta,\kappa},\quad\cP=\cP^\dd=\cQ_{\alpha,\beta},\label{Q4sys}
\end{equation}
where $\cQ_{\alpha,\beta}$ is the polynomial
\begin{equation}
\eqalign{
\cQ_{\alpha,\beta}(u,\wt{u},\wh{u},\wh{\wt{u}}):= \sn(\alpha)\left(u\wt{u}+\wh{u}\wh{\wt{u}}\right)-\sn(\beta)\left(u\wh{u}+\wt{u}\wh{\wt{u}}\right)\\
\qquad{\displaystyle -\sn(\alpha-\beta)\left(u\wh{\wt{u}}+\wt{u}\wh{u}-k\sn(\alpha)\sn(\beta)\left(1+u\wt{u}\wh{u}\wh{\wt{u}}\right)\right).}
}
\label{Q4}
\end{equation}
Thus the polynomials around the cube differ only in the value of parameters, these parameters enter the polynomial through the Jacobi elliptic function $\sn$ with modulus $k$.
This is the Q4 system in its Jacobi parameterisation.
It was originally discovered in \cite{adl} as the superposition principle for B\"acklund transformations of the Krichever-Novikov equation \cite{kn}, the Jacobi parameterisation was obtained in \cite{hie}.

We will consider the Schwarzian variable of the Q4 system with the fixed function $w$ taken to be
\begin{equation}
w=\sqrt{k}\sn(\zeta), \qquad \zeta = \zeta_0 + n\alpha + m\beta.\label{Q4w}
\end{equation}
Here $n,m\in\mathbb{Z}$ are the independent variables, they are incremented by the shifts $\wt{\phantom{u}}$ and $\wh{\phantom{u}}$ respectively.
Note that other solutions of (\ref{wdef}) are possible, for instance choosing $\zeta=\zeta_0+n\alpha-m\beta$ would lead to different results in what follows.
\begin{prop}\label{Q4p1}
If (\ref{gbt}), (\ref{geqs}) is the Q4 system (\ref{Q4sys}) and $w$ is the fixed function given in (\ref{Q4w}), then the following equation holds between a non-singular solution $u$ of the equation on the left in (\ref{geqs}) and its associated Schwarzian variable $\varphi$:
\begin{equation}
\CR{\varphi}{\wt{\varphi}}{\wh{\varphi}}{\th{\varphi}}=
\frac{p\us{q}(1+w\wt{w}s\us{p})(1+\wh{w}\th{w}s\us{p})(\wt{u}-\wt{\us{w}})(\wh{u}-\wh{\ws{w}})}{q\us{p}(1+w\wh{w}s\us{q})(1+\wt{w}\th{w}s\us{q})(\wt{u}-\wt{\ws{w}})(\wh{u}-\wh{\us{w}})},
\label{Q4S}
\end{equation}
in which we have used the notation
\begin{equation*}
\eqalign{
p=\sqrt{k}\sn(\alpha), \quad \us{p}=\sqrt{k}\sn(\alpha-\kappa),\\
q=\sqrt{k}\sn(\beta),\quad\us{q}=\sqrt{k}\sn(\beta-\kappa),\\ 
s=\sqrt{k}\sn(\kappa),\quad \us{w}=\sqrt{k}\sn(\zeta-\kappa), \quad \ws{w}=\sqrt{k}\sn(\zeta+\kappa).
}
\end{equation*}
\end{prop}
\begin{proof}
Substitute the definition of $\mr_1$ and $\mr_2$ into the right-hand-side of (\ref{s2d}).
It is straightforward to verify that this brings (\ref{s2d}) to the form
\begin{equation}
\CR{\varphi}{\wt{\varphi}}{\wh{\varphi}}{\th{\varphi}}=
\frac{q\us{q}}{p\us{p}}\frac{\cQ_{\alpha,\kappa}(u,\wt{u},w,\wt{w})\cQ_{\alpha,\kappa}(\wh{u},\th{u},\wh{w},\th{w})}{\cQ_{\beta,\kappa}(u,\wh{u},w,\wh{w})\cQ_{\beta,\kappa}(\wt{u},\th{u},\wt{w},\th{w})},
\label{QQQQ}
\end{equation}
for {\it any} function $w$, relying only on $\cQ_{\alpha,\beta}(u,\wt{u},\wh{u},\th{u})=0$.
The particular choice (\ref{Q4w}) leads to
\begin{equation}
\eqalign{
\cQ_{\alpha,\kappa}(u,\wt{u},w,\wt{w})=p(1+w\wt{w}s\us{p})(u-\ws{w})(\wt{u}-\wt{\us{w}}),\\
\cQ_{\beta,\kappa}(u,\wh{u},w,\wh{w})=q(1+w\wh{w}s\us{q})(u-\ws{w})(\wh{u}-\wh{\us{w}}),
}
\end{equation}
which together with their once-shifted versions can be substituted into (\ref{QQQQ}) resulting in (\ref{Q4S}).
\end{proof}
{\noindent Thus (\ref{Q4S}) is how (\ref{s2d}) looks in the case of the Q4 system.}
Note that $u$ only appears through values on two vertices of a quadrilateral.

The natural extension of the Q4 system to higher dimensions is obtained by complementing (\ref{Q4sys}) with the associations
\begin{equation}
\cB_3=\cQ_{\gamma,\kappa}, \quad \cP_{23}=\cP_{23}^\dd=\cQ_{\beta,\gamma}, \quad \cP_{31}=\cP_{31}^\dd=\cQ_{\gamma,\alpha}.\label{Q4sys3D}
\end{equation}
Whilst the extended singular solution $w$ is taken to be
\begin{equation}
 w=\sqrt{k}\sn(\zeta), \quad \zeta=\zeta_0+n\alpha+m\beta+l\gamma,\label{Q4w3D}
\end{equation}
where $l\in\mathbb{Z}$ is the third independent variable, which is incremented by the $\wb{\phantom{u}}$ shift.
The higher dimensional consideration of the Schwarzian variable associated with the Q4 system results in the following:
\begin{prop}\label{Q4p2}
If $\varphi$ is the Schwarzian variable associated with the three-dimensional Q4 system (\ref{Q4sys}), (\ref{Q4sys3D}) and fixed function $w$ given in (\ref{Q4w3D}), then it satisfies the equation
\begin{equation}
\frac{(\wh{\varphi}-\th{\varphi})(\wb{\varphi}-\wb{\wh{\varphi}})(\wt{\varphi}-\wt{\wb{\varphi}})}{(\wt{\varphi}-\th{\varphi})(\wh{\varphi}-\wb{\wh{\varphi}})(\wb{\varphi}-\wt{\wb{\varphi}})}=1.\label{SKP}
\end{equation}
\end{prop}
\begin{proof}
Complement (\ref{Q4S}) with similar relations from the other pairs of lattice directions:
\begin{eqnarray}
\CR{\varphi}{\wh{\varphi}}{\wb{\varphi}}{\wb{\wh{\varphi}}}=
\frac{q\us{r}(1+w\wh{w}s\us{q})(1+\wb{w}\wb{\wh{w}}s\us{q})(\wh{u}-\wh{\us{w}})(\wb{u}-\wb{\ws{w}})}{r\us{q}(1+w\wb{w}s\us{r})(1+\wh{w}\wb{\wh{w}}s\us{r})(\wh{u}-\wh{\ws{w}})(\wb{u}-\wb{\us{w}})},\label{Q4SS}\\
\CR{\varphi}{\wb{\varphi}}{\wt{\varphi}}{\wt{\wb{\varphi}}}=
\frac{r\us{p}(1+w\wb{w}s\us{r})(1+\wt{w}\wt{\wb{w}}s\us{r})(\wb{u}-\wb{\us{w}})(\wt{u}-\wt{\ws{w}})}{p\us{r}(1+w\wt{w}s\us{p})(1+\wb{w}\wt{\wb{w}}s\us{p})(\wb{u}-\wb{\ws{w}})(\wt{u}-\wt{\us{w}})}.\label{Q4SSS}
\end{eqnarray}
They can be found by cyclic permutation, we have used the further notation
\begin{equation*}
r=\sqrt{k}\sn(\gamma), \quad \us{r}=\sqrt{k}\sn(\gamma-\kappa).
\end{equation*}
The left-hand-side of the product of (\ref{Q4S}), (\ref{Q4SS}) and (\ref{Q4SSS}) is exactly the left-hand-side of (\ref{SKP}).
The right-hand-side of this product is
\begin{equation}
\frac{(1+\wh{w}\wh{\wt{w}}s\us{p})(1+\wb{w}\wb{\wh{w}}s\us{q})(1+\wt{w}\wt{\wb{w}}s\us{r})}{(1+\wt{w}\wh{\wt{w}}s\us{q})(1+\wh{w}\wb{\wh{w}}s\us{r})(1+\wb{w}\wt{\wb{w}}s\us{p})}
\end{equation}
which is equal to 1. 
This is an elliptic function identity which is a consequence of the addition formula for the Jacobi $\sn$ function.
\end{proof}
The three-dimensional equation (\ref{SKP}) is the lattice Schwarzian KP equation.
It was first identified as the discrete analogue of the Schwarzian KP equation in \cite{dn}, and is gauge-equivalent to the earlier equation given in \cite{ncwq} for generic values of the parameters\footnote{This gauge-equivalence was independently known to G. W. R. Quispel}.
A geometric-incidence interpretation and connection to the Hirota-Miwa equation was established in \cite{ks}.

\section{The degenerate sub-cases of Q4}\label{DEGENS}
The remaining systems classified in \cite{abs1} are all degenerate cases of the Q4 system.
They share the feature that polynomials around the cube differ only in values of parameters, but with the difference that the parameters appear rationally, we denote them by $p$, $q$ and $s$,
\begin{equation}
\cB_1=\cQ_{p,s}, \quad \cB_2=\cQ_{q,s}, \quad \cP=\cP^\dd=\cQ_{p,q}.\label{dsys}
\end{equation}
The polynomials and the names they were given in \cite{abs1} are reproduced here in Table \ref{list}.
\begin{table}[t]
\begin{center}
\begin{tabular}{ll}
\hline
& $\cQ_{p,q}(u,\wt{u},\wh{u},\th{u})$\\
\hline
Q3$^\delta$& 
$(p-1/p)(u\wt{u}+\wh{u}\th{u})-(q-1/q)(u\wh{u}+\wt{u}\th{u})$\\&$ \qquad-(p/q-q/p)[\wt{u}\wh{u}+u\th{u}+\delta^2(p-1/p)(q-1/q)/4]$\\
Q2& 
$p(u-\wh{u})(\wt{u}-\th{u})-q(u-\wt{u})(\wh{u}-\th{u})$\\&$ \qquad+pq(p-q)(u+\wt{u}+\wh{u}+\th{u}-p^2+pq-q^2)$\\
Q1$^\delta$& $p(u-\wh{u})(\wt{u}-\th{u})-q(u-\wt{u})(\wh{u}-\th{u})+\delta^2 pq(p-q)$\\
A2&$(p-1/p)(u\wh{u}+\wt{u}\th{u})-(q-1/q)(u\wt{u}+\wh{u}\th{u}) - (p/q-q/p)(1+u\wt{u}\wh{u}\th{u})$ \\
A1$^\delta$&$p(u+\wh{u})(\wt{u}+\th{u})-q(u+\wt{u})(\wh{u}+\th{u})- \delta^2pq(p-q)$\\
H3$^\delta$& $p(u\wt{u}+\wh{u}\th{u})-q(u\wh{u}+\wt{u}\th{u})+\delta(p^2-q^2)$\\
H2& $(u-\th{u})(\wt{u}-\wh{u})-(p-q)(u+\wt{u}+\wh{u}+\th{u}+p+q)$\\
H1& $(u-\th{u})(\wt{u}-\wh{u})+p-q$\\
\hline
\end{tabular}
\end{center}
\caption{The polynomials listed in \cite{abs1} which define consistent systems on a cube, where it appears $\delta\in\mathbb{C}$ is a constant parameter.}
\label{list}
\end{table}
The `type-Q' (Q3$^\delta$, Q2 and Q1$^\delta$ in Table \ref{list}) polynomials are the most similar to Q4.
For these systems the analog of proposition \ref{Q4p1}, which is proven by a similar calculation, is as follows:
\begin{prop}\label{typeQp}
Suppose the polynomials in (\ref{gbt}) and (\ref{geqs}) are given by (\ref{dsys}) where $\cQ_{p,q}$ is the polynomial $Q3^\delta$, $Q2$ or $Q1^\delta$ listed in Table \ref{list}.
Then the following equation holds between a non-singular solution of the equation on the left in (\ref{geqs}), $u$, and its associated Schwarzian variable, $\varphi$:
\begin{equation}
\CR{\varphi}{\wt{\varphi}}{\wh{\varphi}}{\th{\varphi}}=
\frac{A(p)}{A(q)}\frac{(\wt{u}-\wt{\us{w}})(\wh{u}-\wh{\ws{w}})}{(\wt{u}-\wt{\ws{w}})(\wh{u}-\wh{\us{w}})},
\label{typeQS}
\end{equation}
where the corresponding choice of function $w$, the meaning of notation $\ws{w}$ and $\us{w}$, and the function $A$ are listed in Table \ref{QwA}.
\end{prop}
\begin{table}[t]
\begin{center}
\begin{tabular}{lllll}
\hline
& $w$ & $\ws{w}$ & $\us{w}$ & $A(p)$ \\
\hline
Q3$^\delta$ & $(\zeta+\delta^2/\zeta)/2, \ \zeta=\zeta_0p^nq^m$ & $w(\zeta s)$ & $w(\zeta/s)$ & $(p-1/p)/(p/s-s/p)$\\
Q2 & $\zeta^2, \ \zeta=\zeta_0+np+mq$ & $w(\zeta+s)$ & $w(\zeta-s)$ & $p/(p-s)$\\
Q1$^\delta$ & $w_0+n\delta p+m\delta q$ & $w+\delta s$ & $w-\delta s$ & $p/(p-s)$\\
\hline
\end{tabular}
\end{center}
\caption{Data relevant to Proposition \ref{typeQp}.}
\label{QwA}
\end{table}
{\noindent It is clear from (\ref{typeQS}) that the Schwarzian variable of the Q3$^\delta$, Q2 and Q1$^\delta$ systems extended to three dimensions also satisfies (\ref{SKP}).}

The scenario for systems defined by the remaining polynomials in Table \ref{list} is slightly different.
This is due to the nature of system (\ref{wdef}), which in these cases admits solutions of a quite different character.
However amongst the solutions of system (\ref{wdef}) there do appear constant or oscillating\footnote{By the term oscillating here and subsequently we mean a function which takes only two distinct values depending on whether $n+m$ is even or odd.} functions, and these result in the right-hand-side of (\ref{s2d}) being constant:
\begin{prop}\label{miura1}
Suppose the polynomials in (\ref{gbt}) and (\ref{geqs}) are given by (\ref{dsys}) where $\cQ_{p,q}$ is one of the polynomials listed in Table \ref{list}.
Also suppose the functions $w$ and $A$ are those corresponding to that polynomial in Table \ref{constwA}.
Then the associated Schwarzian variable satisfies
\begin{equation}
\CR{\varphi}{\wt{\varphi}}{\wh{\varphi}}{\th{\varphi}}=\frac{A(p)}{A(q)}.\label{constCR}
\end{equation}
\end{prop}
\begin{table}[t]
\begin{center}
\begin{tabular}{lll}
\hline
& $w$ & $A(p)$ \\
\hline
Q3$^\delta$ & $\infty$ & $(p-1/p)/(p/s-s/p)$\\
Q3$^0$ & $\infty$ or $0$ & $(p-1/p)/(p/s-s/p)$\\
Q2 & $\infty$ & $p/(p-s)$\\
Q1$^\delta$ & $\infty$ & $p/(p-s)$\\
Q1$^0$ & $\vartheta'$ & $p/(p-s)$\\
A2 & $0$ ($n+m$ odd), $\infty$ ($n+m$ even) & $(p-1/p)/(p/s-s/p)$\\
A1$^\delta$ & $\infty$ & $p/(p-s)$\\
A1$^0$ & $(-1)^{n+m}\vartheta$ & $p/(p-s)$\\
H3$^\delta$ & $0$ ($n+m$ odd), $\infty$ ($n+m$ even) & $1/(p^2-s^2)$\\
H3$^\delta$ & $\infty$ & $p/(p/s-s/p)$\\
H3$^0$ & $\infty$ or $0$ & $p/(p/s-s/p)$\\
H2 & $\infty$ & $1/(p-s)$\\
H1 & $\infty$ & $1/(p-s)$\\
\hline
\end{tabular}
\end{center}
\caption{Data relevant to Propositions \ref{miura1} and \ref{miurainv}. Arbitrary constants $\vartheta$ and $\vartheta'$ are chosen from $\ecp\setminus\{0\}$ and $\ecp$ respectively. Note that the singular solutions here are either constant or oscillating functions.  We also remark that the proliferation of $\infty$ as the singular solution is a result of the choice of canonical forms made in \cite{abs1}, a M\"obius change of variables would of course alter this value.}
\label{constwA}
\end{table}
{\noindent Equation (\ref{constCR}) may be written as $\cQ_{A(p),A(q)}(\varphi,\wt{\varphi},\wh{\varphi},\th{\varphi})=0$ where $\cQ_{p,q}$ is the polynomial Q1$^0$ appearing in Table \ref{list}, this is the lattice Schwarzian KdV equation.}
It was first identified as the discrete analog of the Schwarzian KdV equation in \cite{nc} and is a parameter sub-case of the earlier NQC equation \cite{nqc}.
A geometric interpretation connected with discrete conformal maps was given in \cite{bs1}, and the symmetries of this equation were analysed in \cite{lps}.
The reduction from the lattice Schwarzian KP equation to equation (\ref{constCR}) was characterised in terms of the geometric-incidence picture in \cite{ks}, this reduction can also be characterised in terms of the lattice Schwarzian KP B\"acklund transformation as a two-cycle solution, which was shown in \cite{atk3}.

Notice that Proposition \ref{miura1} applies to {\it every} polynomial in Table \ref{list}, the type-Q polynomials recur here but we have made a restricted choice for the function $w$.
The reason for including them again is because of a stronger result which applies: loosely speaking, (\ref{constCR}) actually characterises the co-domain of $\mu$.

The result is most simply formulated by complementing the transformation to the Schwarzian variable described in Proposition \ref{miura1} with explicit equations for the inverse transformation.
As indicated in Section \ref{ODE} it is fairly convenient to proceed by writing 
\begin{equation}
\ms(1)=v_1, \quad \ms(0)=v_2, \quad \ms(\infty)=v_3,
\end{equation}
to deal with the three scalar functions $v_1,v_2,v_3:\mathbb{Z}^2\rightarrow \ecp$ in terms of which the inverse of $\ms$ is given by (\ref{invs}).
The forward transformation to the Schwarzian variable $\varphi$ can then be written
\begin{equation}
\eqalign{
\cB_1(u,\wt{u},v_i,\wt{v}_i)=0,\quad \cB_2(u,\wh{u},v_i,\wh{v}_i)=0, \quad i\in\{1,2,3\},\\
\varphi=\CR{v_1}{v_2}{v_3}{w}.
}\label{str}
\end{equation}
That is, obtaining first three particular solutions $v_1$, $v_2$ and $v_3$ of the B\"acklund system (\ref{gbt}) which corresponds to the first part of system (\ref{str}), and then constructing $\varphi$ rationally from them as in the second part.
Elimination of $v_1$ from (\ref{str}) results in the following explicit system for the inverse transformation:
\begin{eqnarray}
\eqalign{
(v_3-w)\cB_1(u,\wt{u},v_2,\wt{w})\wt{\varphi}-(v_2-w)\cB_1(u,\wt{u},v_3,\wt{w})\varphi = 0,\\
\cB_1(u,\wt{u},v_2,\wt{v}_2) = 0, \quad \cB_1(u,\wt{u},v_3,\wt{v}_3)=0,
}\label{miurap}\\
\eqalign{
(v_3-w)\cB_2(u,\wh{u},v_3,\wh{w})\wh{\varphi}-(v_2-w)\cB_2(u,\wh{u},v_3,\wh{w})\varphi = 0,\\
\cB_2(u,\wh{u},v_2,\wh{v}_2) = 0, \quad \cB_2(u,\wh{u},v_3,\wh{v}_3)=0.
}\label{miuraq}
\end{eqnarray}
\begin{prop}\label{miurainv}
Make again the suppositions of Proposition \ref{miura1}. 
If $\varphi$ is a non-singular solution of (\ref{constCR}) then the systems (\ref{miurap}) and (\ref{miuraq}) define compatible bi-rational mappings $(u,v_2,v_3)\mapsto(\wt{u},\wt{v}_2,\wt{v}_3)$ and $(u,v_2,v_3)\mapsto(\wh{u},\wh{v}_2,\wh{v}_3)$.
If $u$, $v_2$ and $v_3$ are functions determined by these mappings and we define a further function $v$ by the equation
\begin{equation}
\frac{\varphi}{c}=\CR{v}{v_2}{v_3}{w}, \quad c\in\ecp,\label{vv}
\end{equation}
then $(u,v)$ satisfy (\ref{gbt}), (\ref{geqs}), and the Schwarzian variable associated to $u$ is $\varphi$.
\end{prop}
{\noindent Thus system (\ref{miurap}), (\ref{miuraq}) re-constructs $u$ from its Schwarzian variable $\varphi$.}
The proof of Proposition \ref{miurainv} is a straightforward case-by-case calculation.
The non-singularity condition is required for the mappings to be bi-rational.
By examination of (\ref{ndc}) for the lattice Schwarzian KdV equation (\ref{constCR}), the non-singularity just means that $\varphi$ takes distinct values on any pair of adjacent vertices.

The first equations of (\ref{miurap}) and (\ref{miuraq}) are just (\ref{hle}) when $\cR$ is the polynomial defining the Riccati-type B\"acklund equations (\ref{gbt}).
In writing the inverse transformation explicitly we have made several choices.
First, there is some preference in lattice orientation, the mappings described are in the direction of forward lattice shifts, however by inspection one sees that the same system yields mappings in the other directions by a simple re-arrangement (they are bi-rational).
Second, we have chosen to write the transformation as coupled rank-3 mappings (coupled first-order systems), elimination of $v_2$ or $v_2$ {\it and} $v_3$ would lead to rank-2 or scalar equations respectively, but they would be on a larger lattice stencil.
Clearly applying the inverse transformation described yields intermediary functions $v_2$ and $v_3$.
The re-construction of $v_1$ from $v_2$, $v_3$ and $\varphi$ is immediate from the last equality in (\ref{str}), the equation for $v$ (\ref{vv}) is obtained by writing $v=\ms(c)$.

Consider an example.
In the particular case of polynomial H1 in Table \ref{list} the transformation described in Propositions \ref{miura1} and \ref{miurainv} connects the pair of equations
\begin{equation}
\CR{\varphi}{\wt{\varphi}}{\wh{\varphi}}{\th{\varphi}}=\frac{q-s}{p-s}, \qquad (u-\th{u})(\wt{u}-\wh{u})=p-q.\label{example}
\end{equation}
The mappings defined by (\ref{miurap}) and (\ref{miuraq}) for this example are as follows:
\begin{eqnarray}
\fl \qquad \wt{u}=\frac{\wt{\varphi} v_2-\varphi v_3}{\wt{\varphi}-\varphi}, \quad \wt{v}_2=u+\frac{(p-s)(\varphi-\wt{\varphi})}{\varphi(v_2-v_3)}, \quad \wt{v}_3=u+\frac{(p-s)(\varphi-\wt{\varphi})}{\wt{\varphi}(v_2-v_3)},\\
\fl \qquad \wh{u}=\frac{\wh{\varphi} v_2-\varphi v_3}{\wh{\varphi}-\varphi}, \quad \wh{v}_2=u+\frac{(q-s)(\varphi-\wh{\varphi})}{\varphi(v_2-v_3)}, \quad \wh{v}_3=u+\frac{(q-s)(\varphi-\wh{\varphi})}{\wh{\varphi}(v_2-v_3)}.
\end{eqnarray}
(Notice we could take $s=0$ without loss of generality.)
A nice feature of this example is that the equivalence relation induced by $\mu$ (cf. the earlier discussion of property (\ref{muprop})) on the set of non-singular solutions of the equation on the right in (\ref{example}) is also a local symmetry\footnote{We are grateful to S. Butler for pointing out this symmetry, which can be deduced from the soliton solution given in \cite{nah}.}:
\begin{equation}
\fl \qquad u\rightarrow(a+1/a)(u+b)/2+(-1)^{n+m}(a-1/a)(u+c)/2,\quad a\neq 0,b,c\in\mathbb{C}.\label{exsym}
\end{equation}
Thus the transformation to the Schwarzian variable is a bijective correspondence between non-singular solutions of the two equations in (\ref{example}) modulo the M\"obius group and the symmetry group (\ref{exsym}) respectively.

Propositions \ref{miura1} and \ref{miurainv} connect all of the integrable equations classified in \cite{abs1}, with the exception of Q4, to the lattice Schwarzian KdV equation (\ref{constCR}).
By analogy with the situation for continuous KdV-type equations discovered in \cite{ss,ssy}, a transformation between Q4 and the lattice Schwarzian KdV equation should not be expected. 
Here this is connected to the non-existence of constant or oscillating singular solutions of Q4.

As we have described in Section \ref{3D}, the transformation to the Schwarzian variable is compatible with the multidimensional consistency.
In the cases where the transformation connects to the lattice Schwarzian KdV equation, as in Propositions \ref{miura1}, \ref{miurainv} (and Proposition \ref{miura2} later on) this implies commutativity with the natural auto-B\"acklund transformation. 

We remark that in the case of the fifth entry in Table \ref{constwA}, the Propositions \ref{miura1} and \ref{miurainv} define an auto-transformation of the lattice Schwarzian KdV equation.
Note however that $A(p)\neq p$, therefore the transformation connects instances of the same equation, but with different values of its parameters.

\section{Less symmetric systems}\label{THEREST}
We now consider the situation for systems of consistent polynomials of the kind described in Section \ref{ABT} with different symmetry than those in \cite{abs1}.

Several systems where $\cP\neq\cP^\dd$ in (\ref{geqs}) were listed in \cite{atk1}, although the polynomials $\cP$ and $\cP^\dd$ are distinct, each can be taken as one of those listed in Table \ref{list}.
The B\"acklund equations themselves are of the form
\begin{equation}
\cB_1=\cF_p, \quad \cB_2=\cF_q,\label{osys}
\end{equation}
and can be extended to higher dimensions by associating $\cB_3=\cF_r$ and so on for some set of parameters $p,q,r\ldots$.
Some but not all of these systems admit a choice of function $w$ such that the transformation to the Schwarzian variable is invertible in the same way as described in Propositions \ref{miura1} and \ref{miurainv}.
\begin{prop}\label{miura2}
Suppose the polynomials in (\ref{gbt}) are given by (\ref{osys}) where $\cF_p$ is one of the polynomials listed in Table \ref{therest}.
Also suppose the polynomials $\cP$, $\cP^\dd$ and functions $w$ and $A$ are those corresponding to $\cF_p$ in the same table.
Then the Schwarzian variable associated with a solution of the equation on the left in (\ref{geqs}) satisfies (\ref{constCR}).
Conversely, if $\varphi$ is a non-singular solution of (\ref{constCR}) then the systems (\ref{miurap}) and (\ref{miuraq}) define compatible bi-rational mappings $(u,v_2,v_3)\mapsto(\wt{u},\wt{v}_2,\wt{v}_3)$ and $(u,v_2,v_3)\mapsto(\wh{u},\wh{v}_2,\wh{v}_3)$.
If $u$, $v_2$ and $v_3$ are functions determined by these mappings and we define a further function $v$ by (\ref{vv}) then $(u,v)$ satisfy (\ref{gbt}), (\ref{geqs}), and the Schwarzian variable associated to $u$ is $\varphi$.
\end{prop}
\begin{table}[t]
\begin{center}
\begin{tabular}{lllll}
\hline
$\cF_p$ & $\cP$ & $\cP^\dd$ & $w$ & $A(p)$ \\
\hline
$uv+\wt{u}\wt{v}-p(u\wt{v}+v\wt{u})-(p-1/p)(v\wt{v}+\delta^2p/4)$ & Q3$^\delta$ & Q3$^0$ & $0^{\pm 1}$ & $1-p^{\pm 2}$\\
$(v-\wt{v})(u-\wt{u})-p(2v\wt{v}-u-\wt{u})-p^2(v+\wt{v}+p)$ & Q2 & Q1$^1$ & $\infty$ & $p$\\
$(v-\wt{v})(u-\wt{u})-p(v\wt{v}-1)$ & Q1$^1$ & Q1$^0$ & $\vartheta'$ & $p$\\
$(v-\wt{v})(u-\wt{u})-pv\wt{v}$ & Q1$^0$ & Q1$^0$ & $\vartheta$ & $p$\\
$pv\wt{v}-uv-\wt{u}\wt{v}-\delta$ & H3$^\delta$ & H3$^0$ & $\infty$ & $p^2$\\
$(v+\wt{v})(u+\wt{u})-p(v\wt{v}+\delta^2)$ & A1$^\delta$ & A1$^0$ & $(-1)^{n+m}\vartheta$ & $p$\\
$(v+\wt{v})u\wt{u}-p(1-v/2)(1-\wt{v}/2)$ & H3$^1$ & $\!\!^\dagger$A1$^0$ & $(-1)^{n+m}\vartheta$ & $p^2$\\
$(v+\wt{v})(u-\wt{u})-p(v-\wt{v})$ & Q1$^1$ & A1$^0$ & $(-1)^{n+m}\vartheta$ & $p$\\
\hline
\end{tabular}
\end{center}
\caption{Data relevant to Proposition \ref{miura2}. Arbitrary constants $\vartheta$ and $\vartheta'$ are chosen from $\ecp\setminus\{0\}$ and $\ecp$ respectively. $\dagger$ indicates application of the point transformation $p\rightarrow p^2$, $q\rightarrow q^2$ to the parameters. The B\"acklund transformations here are a subset of those given originally in Table 3 of \cite{atk1}.}
\label{therest}
\end{table}
Let us exhibit a simple example: the fifth entry in Table \ref{therest}.
The transformation to the Schwarzian variable in this non-symmetric case is most fully described as connecting the Schwarzian equation to the full system (\ref{gbt}), (\ref{geqs}):
\begin{equation}
\fl\qquad\CR{\varphi}{\wt{\varphi}}{\wh{\varphi}}{\th{\varphi}}=\frac{p^2}{q^2},\qquad
\eqalign{
pv\wt{v}-uv-\wt{u}\wt{v}=\delta, \quad qv\wh{v}-uv-\wh{u}\wh{v}=\delta,\\
p(u\wt{u}+\wh{u}\th{u})-q(u\wh{u}+\wt{u}\th{u})=\delta(q^2-p^2), \\ 
p(v\wt{v}+\wh{v}\th{v})-q(v\wh{v}+\wt{v}\th{v})=0.
}
\label{example2}
\end{equation}
In this instance the mappings described in Proposition \ref{miura2} are as follows:
\begin{equation}
\fl\qquad
\eqalign{
\wt{u}=p\frac{v_2\wt{\varphi}-v_3\varphi}{\wt{\varphi}-\varphi}, \quad \wt{v}_2=\frac{(\wt{\varphi}-\varphi)(u v_2 + \delta)}{p\varphi(v_3-v_2)}, \quad \wt{v}_3 = \frac{(\wt{\varphi}-\varphi)(u v_3 + \delta)}{p\wt{\varphi}(v_3-v_2)},\\
\wh{u}=q\frac{v_2\wh{\varphi}-v_3\varphi}{\wh{\varphi}-\varphi}, \quad \wh{v}_2=\frac{(\wh{\varphi}-\varphi)(u v_2 + \delta)}{q\varphi(v_3-v_2)}, \quad \wh{v}_3 = \frac{(\wh{\varphi}-\varphi)(u v_3 + \delta)}{q\wh{\varphi}(v_3-v_2)}.
}
\end{equation}
Here the partition induced by $\mu$ on the set of non-singular solutions of the system on the right in (\ref{example2}) is not equivalent to a local symmetry.

We remark that for the fourth entry in Table \ref{therest}, Proposition \ref{miura2} defines an auto-transformation for the lattice Schwarzian KdV equation, moreover $A(p)=p$ unlike the transformation of the previous section (of which the transformation here is a degeneration).

Loosely speaking, the fact that the constraint on functions in the co-domain of $\mu$ is two-dimensional is what leads to the existence of the inverse transformation described in propositions \ref{miurainv} and \ref{miura2}.
What tends to happen in more degenerate systems of compatible polynomials is that the transformation to the Schwarzian variable reduces dimensionality, specifically $\varphi$ is constant along some direction on the lattice, which leads to non-invertibility of the transformation.

As an example of this we take a system which looks very much like those listed in Table \ref{therest}. 
In fact we will consider again the fifth entry of Table \ref{therest}, but with the roles of $u$ and $v$ interchanged so that
\begin{equation}
\cF_p=pu\wt{u}-uv-\wt{u}\wt{v}-\delta, \quad \cP={\rm H3}^0, \quad \cP^\dd={\rm H3}^\delta.\label{non-example}
\end{equation}
We take the singular solution $w=\infty$ and again the right-hand-side of (\ref{s2d}) is {\it constant} and thus independent of the solution $u$, so (\ref{s2d}) {\it can} be used to characterise the co-domain of $\mu$. 
But this constant is equal to 1.
Examination of (\ref{s2d}) reveals that this implies either $\wt{\varphi}=\wh{\varphi}$ or $\varphi=\th{\varphi}$, and in fact it turns out that {\it both} are true, thus $\varphi$ is an oscillating function and all information about the original solution $u$ is lost.
Therefore the B\"acklund transformation defined by the fifth entry in table \ref{therest} can in this way be distinguished from its inverse (\ref{non-example}), and in fact the same feature occurs in all other entries of this table except the fourth.

A more obvious degenerate system is the {\it linear} multidimensionally consistent equation defined by the polynomial
\begin{equation}
\fl\qquad \cQ_{p,q}(u,\wt{u},\wh{u},\th{u})=(p_1-q_1)u - (p_2-q_1)\wt{u} - (p_1-q_2)\wh{u} + (p_2-q_2)\th{u},
\end{equation}
whence the compatible system (\ref{gbt}), (\ref{geqs}) is defined by (\ref{dsys}) (cf. \cite{atk2}).
Here the singular solution is $w=\infty$ which leads to a Schwarzian variable that is {\it constant}.

Suffice it to say that for systems listed in \cite{atk1} but which are not included in Table \ref{therest} or gauge-related to systems in Table \ref{constwA}, either $\varphi$ is a constant or oscillating function (like in the examples above), or else (\ref{s2d}) cannot be used to characterise the co-domain of $\mu$ because the right-hand-side depends on $u$ (like in the case of Q4 described in Section \ref{ADLER}).
It may be beneficial to apply the Schwarzian construction to more examples by considering the new systems of consistent polynomials (that also appear to be less symmetric around the cube) which were found recently in \cite{bol}.

\ack
This research was funded by the Australian Research Council Discovery Grant DP 0985615.
We are grateful to F. W. Nijhoff for his comments on the draft manuscript.

\end{document}